\documentclass{article}

\usepackage[utf8]{inputenc} 
\usepackage[T1]{fontenc}    
\usepackage{hyperref}       
\usepackage{url}            
\usepackage{booktabs}       
\usepackage{amsfonts}       
\usepackage{amsmath}
\usepackage{amsthm}
\numberwithin{equation}{section}
\usepackage{fullpage}

\usepackage{bbm} 
\newcommand{\ind}[1]{\mathbbm{1}_{\{#1\}}}

\newtheorem{proposition}{Proposition}[section]

\newtheorem{remark}{Remark}
\newtheorem{example}{Example}[section]
\usepackage{graphicx}
\usepackage{subcaption}

\title{Learning Threshold-Type Investment Strategies with Stochastic Gradient Method}

\author{
  \textbf{Zsolt~Nika} \\
  Faculty of Information Technology and Bionics\\
  Pázmány Péter Catholic University\\
  Budapest, Hungary \\
  \texttt{zsolt.nika@itk.ppke.hu} \\
   \and
 \textbf{Miklós~Rásonyi} \\
  Alfréd Rényi Institute of Mathematics\\
   Hungarian Academy of Sciences\\ Budapest, Hungary \\
}
\date{June 2019}

\begin{document}

\maketitle

\begin{abstract}
    
    In online portfolio optimization the investor makes decisions based on new, continuously incoming information on financial assets (typically their prices).
    In our study we consider a learning algorithm, namely the Kiefer--Wolfowitz version of the Stochastic Gradient method,
    that converges to the log-optimal solution in the threshold-type, buy-and-sell strategy class.
    
    The systematic study of this method is novel in the field of portfolio optimization; we aim to establish the theory and practice of Stochastic Gradient algorithm used on parametrized trading strategies.
    
    We demonstrate on a wide variety of stock price dynamics (e.g. with stochastic volatility and long-memory) 
    that there is an optimal threshold type strategy which can be learned. 
    Subsequently, we numerically show the convergence of the algorithm.
    Furthermore, we deal with the typically problematic question of how to choose the hyperparameters 
    (the parameters of the algorithm and not the dynamics of the prices) 
    without knowing anything about the price other than a small sample.
   
\end{abstract}

\paragraph{Keywords}
Stochastic gradient; Log-optimal investment; Online portfolio selection

\section{Introduction}




    In investment there is an approach by technical analysts where the investment decision is based on past data 
    such as price, technical indicators or trading volumes. 
    The decisions are determined by some function of past data called \textit{trading rule} or \textit{strategy function}.
    In algorithmic trading, these decisions are executed automatically by computers \cite{kim2010electronic}.
    One of the most typical subcategories of algorithmic trading is \textit{high-frequency trading}, where favorable decisions must be made
    in seconds or even miliseconds \cite{aldridge2013high}.
    
    Given the nature of these algorithms, those that require huge computational capacities are not efficient since they are slow.
    For example, non-parametric methods or complex machine learning algorithms work well on big data sets with immense computer efforts (see a survey of non-parametric methods e.g. in \cite{li2014online} and a summary in machine learning methods \cite{de2018advances}).
    On the other hand, parametric models that are based on dynamics of the prices or indices may give fairly good results if precise and accurate parameter estimations are available.
    To get satisfactory estimations, again one needs a big data set, and typically decisions are sensitive to the error of the estimations.
    
    To resolve the above-mentioned problems we use Kiefer--Wolfowitz method which lets us 
    (i) make decisions immediately starting at the initial step;
    (ii) process new information/data as they arrive, without needing to wait until we have a big enough data set, as the strategy function improves in every step and 
    (iii) there is no need to estimate the parameters of the dynamics.
    With this method, we aim to optimize \textit{log-utility} investments (maximizing the expected value of the logarithm of the wealth).
    The method is also capable of tracking the changes of the market, 
    which we ignore here in order to investigate the method itself in finance instead of market changes.

    Stochastic Approximation \cite{robbins1951stochastic} (or Robinson--Monro method) is an iterative method to find the root of $f(\theta):=\mathbb{E}[X_t, \theta]=m$, where $X_t$ is a stochastic process, $\theta$ is a prameter and $m$ is a constant.
    Basically, it is a stochastic version of the Newton--Raphson method where there are consecutive observations of the functions loaded with randomness/noise.
    If the derivate of the function exists then the method can be used for optimization.
    When the derivate does not exist or is unknown, Kiefer and Wolfowitz proposed \cite{kiefer1952stochastic} a finite difference approximation based on consecutive observations. This method is a version of the Stochastic Gradient method.
    
    Nonetheless, stochastic approximation or Kiefer--Wolfowitz has not been used for directly optimizing the parametrized strategy as we do here.
    We hope that this introduction to the usage of the method in investment theory will develop further.
    Other works dealt with different approaches, like
    learning the parametrized stopping time for American/Asian options \cite{zhang2011stochastic} or the optimal stopping time of liquidation \cite{yin2006stock}. 
    Other typical fields of applying this algorithm are the estimation of quantiles for CVaRs \cite{chow2014algorithms}, 
    \cite{laruelle2012stochastic}or \cite{kibzun2001discrete}.
    There is also a study about the optimal splitting of orders in \cite{laruelle2011optimal}.
    
    In Section \ref{sec:threshold} we introduce the threshold strategies that can be parametrized in a way that it can be optimized by the Kiefer--Wolfowitz algorithm.
    Then in Section \ref{sec:KW} we show how the algorithm works and that the optimum exists.
    In Section \ref{sec:numerical} we show numerical results how the algorithm performs and we also deal with the problem how to choose the hyperparameters of the algorithm in a suitable way. 
    Throughout this article we make some usual simplifications in investment theory:
    the investment solely contains one risky and one riskless asset,
    of course, these can relaxed.
    Given the nature of the learning method, we only focus on discrete time models.

\section{Threshold strategies in log-optimal investments}\label{sec:threshold}

    In this section, we introduce the financial background in which we want to apply the learning method in the next sections.
    At first, we start with some preliminary information about investment, after which we present our threshold type strategy
    and we discuss how it connects to the theory of log-optimal investments.
    
    \subsection{Portfolio}
    Portfolio investment, mathematically speaking, is an applied field of control theory where the control process
    is the investor's decision regarding in which asset to allocate her/his current wealth, and the independent processes
    are typically the prices of the financial assets.
    Logarithmic utility function is used frequently as an objective function for several reasons.
    
    Let us denote the riskless process $B_t$ and the risky asset $S_t$, where $t\in\mathbb{N}$ is the discrete-time parameter.
    In this study, we do not want to focus on the effect of the interest rate, therefore, the riskless process
    is chosen to be constant (it assumes zero interest rate, so we do not need to discount the prices).
    The value of the portfolio $W_t$ is the wealth of the invester and its time-evolution is typically written as
    \[
    	\frac{W_t}{W_{t-1}} = (1-\pi_t)\frac{B_t}{B_{t-1}} + \pi_t \frac{S_t}{S_{t-1}},
    \]
    where $\pi_t\in[0,1]$ is an $\mathcal{F}_{t-1}$-measurable function, called \textit{strategy}, i.e. 
    the fraction of how much of the current wealth should be split between the two assets.
    Clearly, $\pi_t$ can only be a function of information up to $t-1$ 
    since the investor is not able to look in the future.
    In financial mathematics, the log increment of the price has several well-established
    properties mathematically, which are called the stylized facts of stock prices \cite{cont2001empirical}. It is convenient to
    build a dynamics on the log-return and not on the stock price.
    The log-return
    \[
    	H_t:= \log\left( \frac{S_t}{S_{t-1}} \right).
    \]
    Since the riskless asset's price is constant (therefore their fraction is one) we can simplify the wealth as
     \begin{equation}\label{eq:growth}
    	\frac{W_t}{W_{t-1}} = 1-\pi_t + \pi_t e^{H_t}.
    \end{equation}
    
    The investor's objective is to maximize the utility function
    \begin{equation}
    	\lim_{t\rightarrow\infty}\frac{1}{t}\mathbb{E}[\log(W_t)].
    \end{equation}
    
    It has been showed in \cite{algoet1988asymptotic} that it can be maximized if the strategy $\pi_t$ is chosen 
    such as to maximize the conditional expectation of the growth
    \begin{equation*}
    	\tilde{g}:=\mathbb{E}[\log(W_t/W_{t-1})|\mathcal{F}_{t-1}] \quad\rightarrow \quad\text{maximize},
    \end{equation*}
    or with our financial conditions it equals to
    \begin{equation}\label{logopt}
    	\tilde{g}:=\mathbb{E}[\log(1-\pi_t + \pi_t e^{H_t})|\mathcal{F}_{t-1}] \quad\rightarrow \quad\text{maximize}.
    \end{equation}
    The conditional expectation $\tilde{g}$ is a random variable and measurable on $\mathcal{F}_{t-1}$.
    The condition on $\mathcal{F}_{t-1}$ contains a lot more information that is accessible for an investor or anyone.
    In an algorithm we need to specify what information we use (for example past prices or stock market indices)
    therefore we can only optimize a conditional mean where the condition is a random variable.
    We denote it as
    \begin{equation}
        \tilde{g}(X) = \mathbb{E}[\log(1-\pi_t + \pi_t e^{H_t})|X],
    \end{equation}
    where $X$ is an $\mathcal{F}_{t-1}$-measurable (multivariate) random variable.
    
    In the following sections we show how to parametrize the strategy process $\pi_t$ to be able to learn the log-optimal strategy and then how to choose the variable $X$.

    \subsection{Dynamics}
    The present method can be used on several type of stock price dynamics.
    It is important to use such dynamics
    where (i) the optimal strategy exists and results in a portfolio which achieves its optimality
    and (ii) the price dynamics is realistic, plausible.
    For this reason we rely on the time series class introduced in \cite{nika2018log} called 
    \textit{Conditionally Gaussian} and one of its example, the 
    \textit{Discrete Gaussian Stochastic Volatility} (DGSV):
    \begin{subequations}\label{eq:dgsv}
	\begin{equation}
	    H_t = \mu + \alpha H_{t-1} + \sigma e^{Y_t}\left(
	\rho\varepsilon_t + \sqrt{1-\rho^2}\eta_t\right),
	\end{equation}
	\begin{equation}
	    Y_t=\sum\limits_{j=0}^{\infty} \beta_{j} \varepsilon_{t-j},
	\qquad \beta_{j},\mu,\sigma\in\mathbb{R};\,  \alpha,\rho\in[-1,1].
	\end{equation}
    \end{subequations}
    This stock price model posseses several desirable properties:
    its statistical moments and auto-correlation function are realistic,
    includes long-memory and leverage effect as well.
    The existence of the log-optimal solution is provided in \cite{nika2018log}.
    
    We also use simplier models to understand better the behavior of the algorithm.
    Such as AR(1) or MA($\infty$) processes:
	\begin{equation}\label{eq:ar1}
	    H_t = \mu + \alpha H_{t-1} + \sigma \varepsilon_t \quad\text{: AR(1)},
	\end{equation}
	\begin{equation}\label{eq:ma}
	    H_t=\mu + \sum\limits_{j=0}^{\infty} \beta_{j} \varepsilon_{t-j}\quad\text{: MA($\infty$)}.
	\end{equation}
    The coefficients $\beta_j:=b_0(1+j)^{-b}$ and choosing $b_0>0$ and $0.5<b<1$ ensure that \ref{eq:ma} has long memory.
    
    \subsection{Threshold strategy}
    The log-optimal strategy in (\ref{logopt}) only can be calculated if the parameters of the stock price dynamics are known. 
    The exact form of the strategy in unknown, one need to use numerical integration to get the optimal decision at every 
    timestep $t$.

    In Section 3 of \cite{nika2018log} an approximative strategy of the log-optimal was proposed.
    They showed that on realistic data it performs well, though
    they did not give mathematical estimation of the error.
    This approximative strategy reduces the space of possible decisions from $\pi_t \in[0,1]$ to two states $\pi \in\{0,1\}$.
    With realistic log-return data this restriction does not result in a considerable loss and
    it can be used with learning algorithms while the log-optimal solution can't.
    
    The idea can be used for any parametric dynamics if the conditional expectation can be calculated.
    The proposed approximative strategy in \cite{nika2018log} is
    \begin{equation}\label{eq:lin}
        \pi_t^{lin} = \begin{cases}
        1, \quad \text{if } \mathbb{E}[H_t | \mathcal{F}_{t-1}] > 0,\\
        0, \quad \text{otherwise,}
        \end{cases}
    \end{equation}
    which is a consequence of the requirement in (\ref{logopt}) with first-order Taylor-expansion.
    That is, the investor should buy only risky asset if its conditional expected value is higher than 0.
    This strategy lies in the field of threshold strategy.
    
    We remark, that we are working now in 0 interest rate environment.
    Without this assumption the trading rule modifies to buy whenever the conditional expectation is higher than
    the interest rate.
    
    Because of the structure of the strategy, we call it here \textbf{threshold} strategy.
    We do not need the upperscript $lin$ since we are only investigating this type of strategy with the Stochastic Gradient method.
    
    In most parametric models the conditional expectation can be calculated therefore we end up with a function of past data that we call here
    \textbf{threshold-function}:
    $f(\text{past data}):=\mathbb{E}[H_t | \text{past data}]$. 
    An equivalent form of \ref{eq:lin} using the threshold function is
    \begin{equation}
        \pi_t = \ind{ f(\text{past data}) > 0 },
    \end{equation}
    where the function $\mathbf{1}_{\{x>0\}}$ is 1 if $x>0$ and 0 otherwise.
    The conditional expectation of the growth (\ref{logopt}) that we want to optimize here with the indicator function is
    \begin{equation}
        \tilde{g} = \mathbb{E}[\log(1 - \ind{f(\text{past data}) > 0} + 
        \ind{f(\text{past data}) > 0)}e^{H_t}|\text{past data}].
    \end{equation}
    This function is still a random variable because it is a function of past data.
    
    In the following subsections we unfold some cases how to handle "past data", but
    of course, it is the investors duty to tell, which past values to use.
    Proposition \ref{th:markovopt} gives help how and what to take into consideration when someone chooses values
    from past data.
    
    With Stochastic Gradient method we are able to optimize an expected value with respect to some parameters.
    Therefore in the following we will optimize the exptected value of $\tilde{g}$.
    If we parametrize the conditional growth by $\theta$ which is a one or multivariate real number, than the optimization task is to find the maximum of the \textbf{growth}
    \begin{equation}\label{eq:param-growth}
        g(\theta) := \mathbf{E}[\tilde{g}(X_{t-1}, \theta)],
    \end{equation}
    where $\tilde{g}(X_{t-1}, \theta)$ is a parametrized version of (\ref{logopt}).
    
    \subsection{Markovian strategy}
    Let us assume the investor uses only one value that is available before investing at time $t$ and call this variable $X_{t-1}$. 
    It can be past stock returns or an index or something more complex, for example the weighted average of the past returns.
    A natural choice can be the previous value of the return, that is $H_{t-1}$ and we stick to this simple case here.
    
    The conditional growth in (\ref{logopt}) with Markovian strategy:
    \begin{equation}\label{eq:markovgrowth}
        \mathbb{E}[\log(1-\pi_t+\pi_te^{H_t})|X_{t-1}].
    \end{equation}
    We need to parametrize the threshold function in the strategy to be able to use it with Stochastic Gradient method.
    A convenient choice is the linear function; in this paper we do not relieve this restriction but we mention that $X_{t-1}$ can be a function of $H_{t-1}$ though.
    \begin{equation}\label{eq:markov-strat}
        \pi_t:=\ind{X_{t-1}>\theta}.
    \end{equation}
    The optimizable  growth in (\ref{eq:param-growth}) is
    \begin{equation}\label{eq:markovgrowth2}
        g(\theta) =  \mathbb{E}[H_t \ind{X_{t-1}>\theta}].
    \end{equation}
    In Section \ref{sec:KW} we will optimize this function
    with the Kiefwer--Wolfowitz method.
    
    The theorem below shows the optimal threshold ($\theta^*$) of the Markovian strategy.
    \begin{proposition}\label{th:markovopt}
        Let us assume that there is only one root of the differentiable function
        $\phi(x):= {\mathbf{E}[H_t|X_{t-1}=x]}$ and that $\phi(x)>0$ if $x>0$.
        Moreover, let us assume that the return process is stationary.
        Then the root of $\phi(x)$ is the unique optimal threshold:
        \begin{equation}
            \theta^* = \{x|\phi(x)=0\}.
        \end{equation}
    \end{proposition}
    
    \begin{proof}
        For the sake of simplicity assume that $X_{t-1}$ has a pdf.
        The conditional expectation of the growth is
        \begin{equation*}
            g(\theta) =  \mathbb{E}[H_t \ind{X_{t-1}>\theta}] = \mathbb{E}\left[\mathbb{E}[H_t|X_{t-1}]\ind{X_{t-1}>\theta}  \right].
        \end{equation*}
        Since $\mathbb{E}[H_t|X_{t-1}]$ is a function of $X_{t-1}$, call it $v(X_{t-1})$ and denote the pdf of $X_{t-1}$ as $f_X(x)$, the expected value is
        \begin{equation*}
            \int_\theta^\infty v(y) f_X(y) dy.
        \end{equation*}
        The integral has optimum where
        \[
            -v(y)f_X(y) = 0.
        \]
        Since $f_X(y)$ is non-negative therefore the optimal threshold is where
        $v(y) = 0$ which conclude our statement.
        
    \end{proof}
    
    \begin{remark}\label{rem:meanindependent}
        The main message of the theorem is that 
        only those information can be used in
        the optimization algorithm which are not \textit{mean-independent} \cite{cameron2005microeconometrics} from the
        price process. The concept of mean independence is well-known in econometrics
        which is a stronger property than uncorrelation but weaker than the stochastic independence. 
    \end{remark}
    
    \begin{remark}
        A conclusion of Proposition \ref{th:markovopt} is that the linear approximative strategy is log-optimal if the strategy can only be 0 or 1. This is only true in the univariate case.
    \end{remark}
    
    For a simple example, let us model the log-return as an autoregressive process
    and let us use the previous log-return value as "past data".
    \begin{example}[AR(1)]
        Let $H_t$ defined as in (\ref{eq:ar1}).
        The conditional expectation is
        \begin{equation*}
            \mathbf{E}[H_t|H_{t-1}=x] = \mu + \alpha x.
        \end{equation*}
        Its root, that is the optimal threshold is
        \begin{equation*}
            \theta^* =-\frac{\mu}{\alpha}.
        \end{equation*}
        When $\alpha<0$, then the assumption in Theorem \ref{th:markovopt} about $\phi(x)>0$ if $x>0$ is false, but the optimality is true if we change the inequality sign in
        (\ref{eq:markov-strat}) to $\pi_t:=\ind{X_{t-1}<\theta}$.
        
        (We remind the reader, that the expected log-return is different, $\mu/(1-\alpha)$.)
    \end{example}
    As we can see from the example, to determine the threshold we either need to estimate $\mu$ and $\alpha$ from a long enough sample or either we learn the value of $\theta^*$
    by using Stochastic Gradient. In a more realistic dynamics there are more than two
    parameters that needed to be estimated. Furthermore the threshold is very sensitive
    to the estimation error of $\alpha$.
    \begin{example}[DGSV]
        Let the log-return $H_t$ be a DGSV process according to (\ref{eq:dgsv}).
        Its conditional expectation is
        \begin{equation*}
            \mathbf{E}[H_t|H_{t-1}=x] = \mu + \alpha x + \sigma\rho\mathbb{E}[e^Y_t|H_{t-1}].
        \end{equation*}
        The conditional expectation is unknown but we will see later in the numerical results that there is a unique solution.
    \end{example}

    \subsection{Non-Markovian strategy - multivariate case}
    If the investor rather would like to use more information for example to handle
    long memory or information about volatility, it is also possible.
    We show here two possible choices that can be used, one strategy uses multiple
    past return data, the other one uses volatility information as extra.
    The strategies
    \begin{subequations}
    \begin{equation}\label{eq:sratmulti}
        \pi_t=\ind{H_{t-1}+\theta^2H_{t-3}+\theta^3H_{t-3}+\dots>\theta^1}\quad
        \text{or}
    \end{equation}    
    \begin{equation}\label{eq:stratdgsv}
        \pi_t=\ind{H_{t-1}+\theta^2 e^{\nu_{t-1}}>\theta^1},
    \end{equation}
    \end{subequations}
    where $\theta^1,\theta^2,\dots$ are the parameters we wish to optimize and
    $\nu_{t-1}$ is an estimation of the logarithm of the volatility based on the information of  $\mathcal{F}_{t-1}$ (that is, $\nu_{t-1}:=\mathbb{E}[Y_t|\mathcal{F}_{t-1}]$).
    The design of the second strategy with the log-volatility may seem peculiar but the linear approximation strategy of the log-optimal in \cite{nika2018log} has been showed that it is a linear function of $H_{t-1}$
    and $\nu_{t-1}$.

An important aspect of the strategy choice with volatility, is that we are able to catch leverage effect with it. As we noted in Remark \ref{rem:meanindependent}, only those processes should be used in the threshold function which are not mean-indepenent of the log-return. Leverage effect is defined in several ways, anyhow it is a connection between stock price change and past volatility (i.e. in our case between $H_{t}$ and $\nu_{t-1}$). 
Noises in the price that have no leverage effect, for example the noise term $\eta_t$ in \ref{eq:dgsv}, have no advantagesin the investment.

    Leverage effect has a prominent role, since it is the only way how we can utilize volatility but the long memory 
    typically appears in volatility.
    As it has been show in \cite{cont2001empirical}, the long memory is hidden in volatility and not in the drift part of the process.
    
    In the multivariate case there is no closed form of the optimal $\theta^i$ values.
    Of course, the $\partial g/\partial \theta^i = 0$ must be satisfied.
    For example, in two dimensions version of (\ref{eq:sratmulti}) the optimal $\theta$'s must satisfy the
    \begin{subequations}
    \begin{equation}
        \partial g/\partial \theta^1 =\int_{-\infty}^\infty  v(\theta^1-\theta^2x,x)f(\theta^1-\theta^2x,x) dx = 0,
    \end{equation}
    \begin{equation}
        \partial g/\partial \theta^2 =\int_{-\infty}^\infty - x v(\theta^1-\theta^2x,x)f(\theta^1-\theta^2x,x) dx = 0,
    \end{equation}
    \end{subequations}
    equations, where $v(x,y):=\mathbb{E}[H_{t} |H_{t-1}=x, H_{t-2}=y ]$ and $f(x,y)$ is the joint pdf of $(H_{t-1}, H_{t-2})$.
    The equations are more complicated in the DGSV case if we wish to include the log-volatility $\nu_{t-1}$
    then we need to replace the variable $x\rightarrow \exp(x)$ and reinterpret the pdf and conditional mean (by using $e^{\nu_{t-1}}$ instead of $H_{t-2}$.
    These are unknown functions in general and we could only estimate the pdf and the conditional expectation
    based on data which is contrary to our goals.
    
    It does not mean that the Kiefer--Wolfowitz algorithm cannot converge to the optimal $\theta$'s, only
    that we cannot calculate their optimal values in advance. If the dynamics are known then Monte-Carlo method can be used to estimate the optimal value. This is what we use in the numerical simulations.
    
    Here we would like to show the basics of how to use the Kiefer--Wolfowitz algorithm for investment purposes.
    Other processes could also be used.
    
\section{Kiefer--Wolfowitz algorithm}\label{sec:KW}
    With the Kiefer-Wolfowitz optimization procedure we are searching for the maximum of 
    (\ref{eq:param-growth}).
    \paragraph{Univariate case:} the task is to find the optimum threshold $\theta^*\in\mathbb{R}$
    \begin{equation}
        \text{maximize}_{\theta} \quad g(\theta) := \mathbb{E}[H_t\ind{X_{t-1}>\theta}],
    \end{equation}
    the random processes $H_t$ and $X_{t-1}$ are both univariate.
    Let us denote the growth at time $t$ by $G(\theta; H_t, X_{t-1}):=H_t\ind{X_{t-1}>\theta}$.
    The Stochastic Gradient algorithm uses the finite differences of the growth:
    \begin{equation}
        \theta_{t+1} = \theta_{t} + a_t\frac{G(\theta_t-c_t; H_t, X_{t-1}) - G(\theta_t+c_t; H_t, X_{t-1})}{c_t},
    \end{equation}
    where the step-size $a_t$ and the step-size of the finite difference $c_t$ are real-valued sequences.
    The fraction is the approximation of the gradient.
    
    Since the growth $G(\theta;\dots)$ is the indicator function of $\theta$, therefore its finite difference can be simplified to a range.
    For greater clarity we denote the range $[x-c, x+c]$ as $[x\pm c]$.
    Then the algorithm can be written as
    \begin{equation}\label{eq:range}
        \theta_{t+1} = \theta_{t} + a_t \frac{H_t\ind{X_{t-1}\in[\theta_t\pm c_t]}}{c_t}.
    \end{equation}
    This formalism will help us in the latter to better understand the usage of the method.
    
    It is impossible to prove in general but via some examples in the Section \ref{sec:numerical}
    we show nuerically that this recursive update converges to the optimum what we showed in the previous section:
    \begin{equation}
        \theta_t \xrightarrow{L^2} \theta^*,
    \end{equation}
    the convergence is in $L^2$, i.e. we can show the convergence of the Mean Squared Error (MSE).
    If the convergence is accomplished, its speed has power-law typically.
    
    In general, there is no straightforward way to choose the hyperparameters.
    In Section \ref{sec:numerical} we show some ideas 
    on which basis we can choose the hyperparameters.
    \smallskip
    
    In real life investment the financial environment is not static, the dynamics of prices can change and new factors can appear/disappear, therefore optimal strategy changes as well.
    To this end, in practice investors use constant and very small step sizes $a_t$ and $c_t$
    which able to track down the changes of the optimal values.
    In this paper we do not aim to focus on changes of the market.
    
    \paragraph{Multivariate case:} the algorithm works in the same way,
    each dimension of the parameter are updated separatly with no cross-effect.
    For example in the case of known log-volatility (\ref{eq:stratdgsv}) the growth is ${G(\theta^1, \theta^2; H_t, H_{t-1}, \nu_{t-1})}$
    \begin{subequations}
    \begin{equation}
        \theta^1_{t+1} = \theta^1_{t} + a_t^1 \frac{H_t\ind{H_{t-1}\in[\theta_t^1-\theta_t^2 e^{\nu_{t-1}}\pm c^1_t]}}{c^1_t}
    \end{equation}
    \begin{equation}
        \theta^2_{t+1} = \theta^2_{t} + a^2_t \frac{H_t\ind{H_{t-1}\in[\theta_t^1-\theta_t^2 e^{\nu_{t-1}}\pm c^2_te^{\nu_{t-1}}]}}{c^2_t}
    \end{equation}
    \end{subequations}


\section{Numerical Results}\label{sec:numerical}
    The critical part of every algorithm is the choice of the hyperparameters.
    In their paper, J. Kiefer and J. Wolfowitz \cite{kiefer1952stochastic} also address the issue 
    of parameter-choice though they were able to give exact and 
    sufficient conditions in a simplier context.
    These conditions are typical requirements and our model satisfy them as well:
    \begin{enumerate}
        \item $c_t \rightarrow 0$.
        \item $\sum_{t=1}^\infty a_t = \infty$, that is, the algorithm can reach any state.
        \item $\sum_{t=1}^\infty a_t c_t < \infty$.
        \item $\sum_{t=1}^\infty a_t^2 c_t^{-2} < \infty$.
    \end{enumerate}
    A usual first guess choice is $a_t=t^{-1}$ and $c_t=t^{-1/3}$.
    
    Analyzing the growth function $g(\theta)$ in the univariate case help us to construct the step-sizes in a suitable way.
    Figure \ref{fig:hill} and (\ref{eq:range}) make it clear that $\theta_t$ must stay in the same range as $X_{t-1}$, since $X_{t-1}\not\in [\theta_t \pm c_t],\, \forall t\in\mathbb{N}$ would result in constant $\theta_t$. In the numerical simulations we only show results about the $X_{t-1}:=H_{t-1}$ case.    
    On the two example we can make the following remarks:
    \begin{itemize}
        \item $g(\theta\rightarrow -\infty) = \mathbb{E}[H_t]$, low $\theta$ means that $\pi_t=1$, that is, the wealth equals to the price of the stock.
        \item $g(\theta\rightarrow \infty) = 0 $, high $\theta$ means
        that $\pi=0$, the wealth equals to the price of the bond.
        \item In the simple case when $H_t$ is an autoregressive process and also when it has the more complex, realistic dynamics DGSV, there is a unique $\theta^*$ that can be calculated.
        \item If $\theta_t$ takes value out of the typical value of $H_t$ where the derivate of $g(\theta)$ is zero then it is hopeless for the algorithm to return and it stays there.
    \end{itemize}
    
    To overcome on the problem of the last remark we make some modifications on the algorithm.
    First, the inital value $\theta_0$ must be estimated on a small sample of $H_t$.
    In every realization we used 10 data points to initialize $\theta_0:=\sum_{t=1}^10 H_t/10$.
    This very small sample is already enough for the algorithm to start from a relatively good point.
    Second, we cannot let the algorithm to take any large step. A general solution for this is to use a project $\theta_t$ on a subspace. In our case we do a truncation on the known range of $H_t$:
    \[
        \theta_t = \begin{cases}
        \min_{1\leq j \leq t} H_j , \quad \text{if } \tilde{\theta}_t <  \min_{1\leq j \leq t} H_j,\\
        \tilde{\theta}_t, \quad \text{otherwise},\\
        \max_{1\leq j \leq t} H_j , \quad \text{if } \tilde{\theta}_t >  \max_{1\leq j \leq t} H_j,
        \end{cases}
    \]
    where $\tilde{\theta}:= \theta_{t-1} + a_t H_t\ind{H_{t-1}\in[\theta_{t-1}\pm c_{t-1}]}/c_{t-1}$.

    \begin{figure}
        \centering
        \begin{subfigure}{.5\textwidth}
          \centering
          \includegraphics[width=0.8\linewidth]{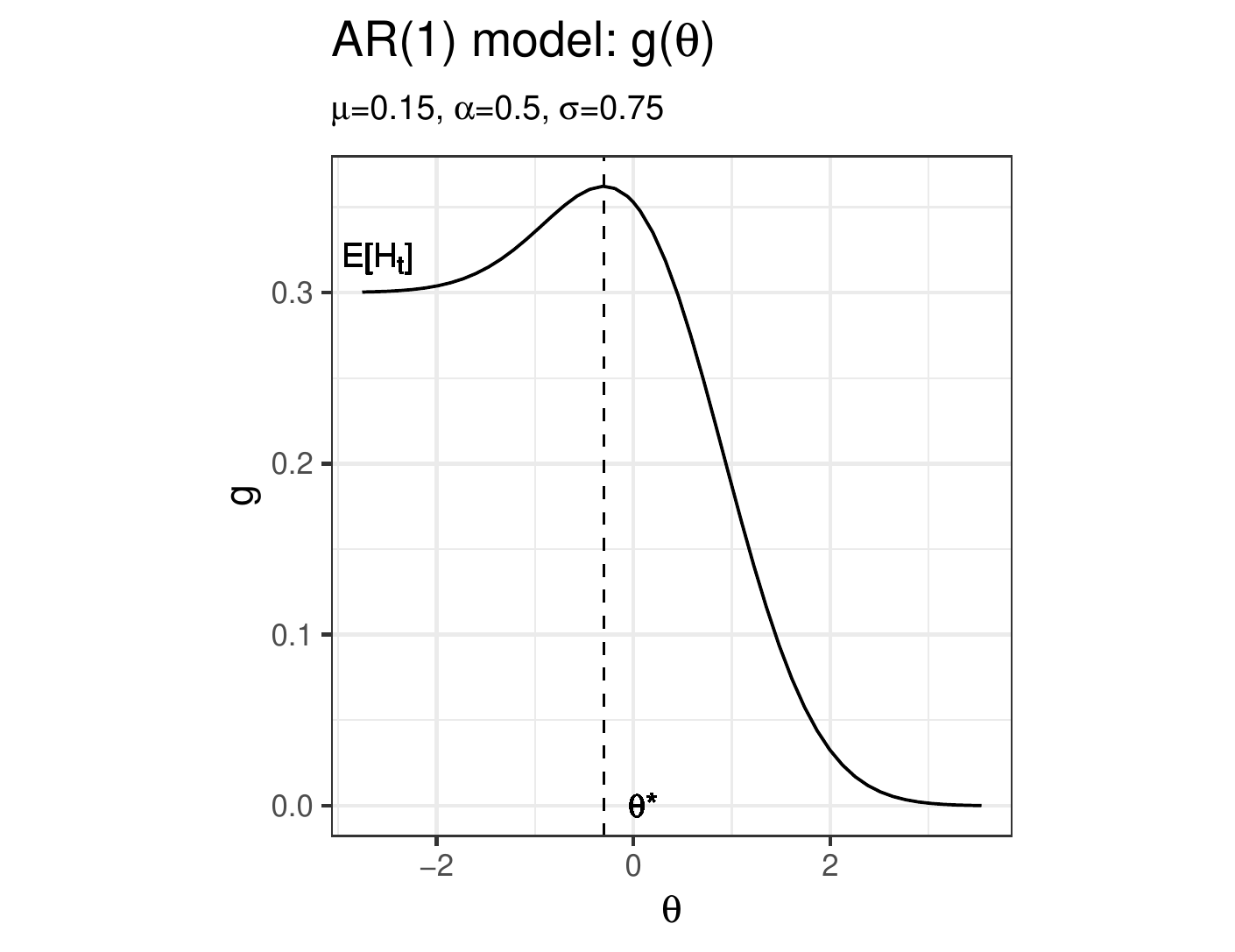}
          \caption{Log-return is an AR(1) process}
          \label{fig:hillar1}
        \end{subfigure}%
        \begin{subfigure}{.5\textwidth}
          \centering
          \includegraphics[width=0.8\linewidth]{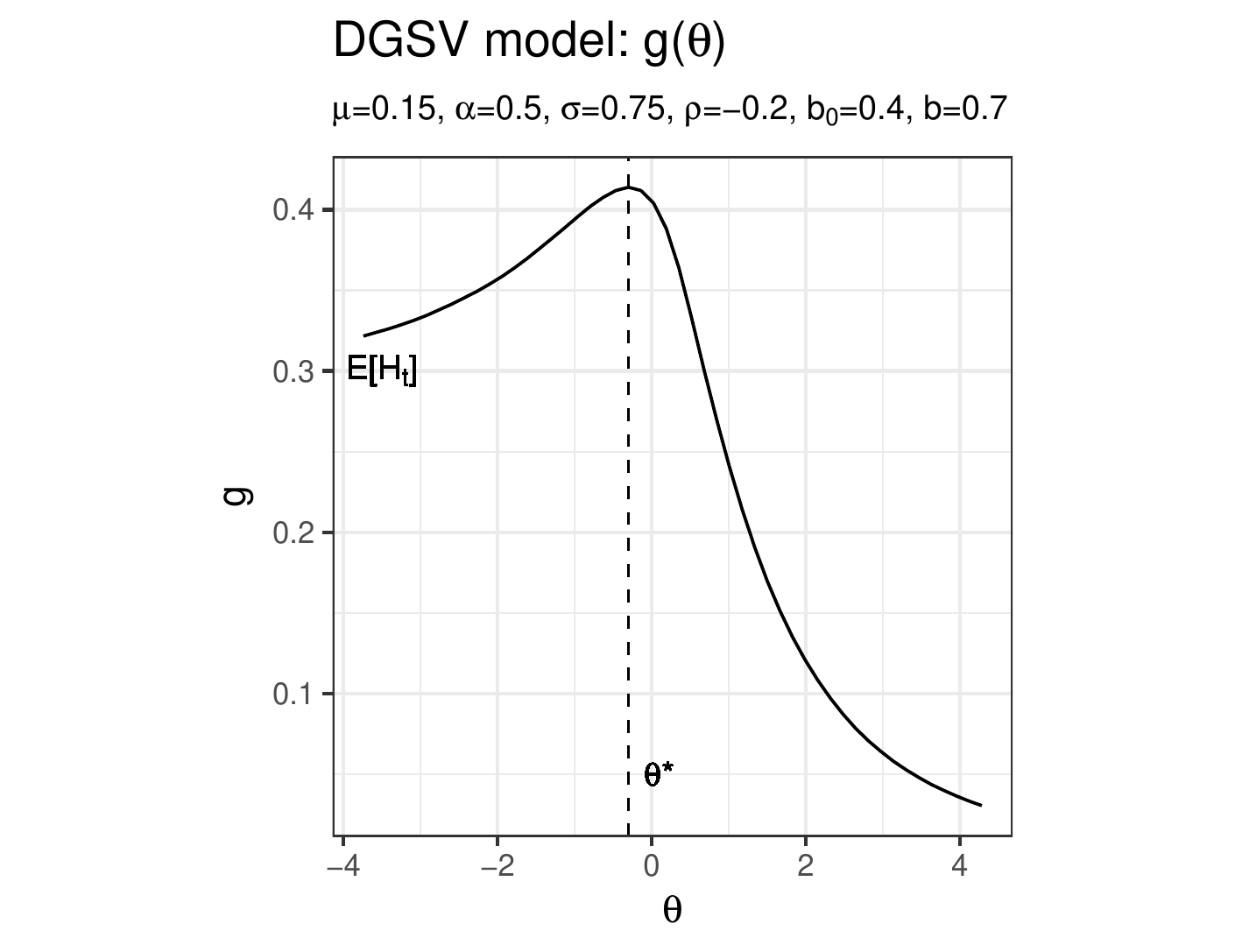}
          \caption{Log-return is a DGSV process}
          \label{fig:hilldgsv}
        \end{subfigure}
        \caption{Function $\theta \rightarrow g(\theta)$. Both plot shows $\theta$ values on the 0.01 and 0.99 percentile range of $H_t$.}
        \label{fig:hill}
    \end{figure}
    
    Using the simple parametrization $a_t=t^{-1}$ and $c_t=t^{-1/3}$ can work in a simple setting.
    Figure \ref{fig:kw} show hot the simple choice of the hyperparameters work.
    The simulations were executed with $N=25$ realizations and for $T=50\,000$ time steps.
    The Mean Squared Error (MSE) is an approximation of the $L^2$ error.
    The log-log scale plot of the error shows that MSE has power law decays in both cases.
    
    The requirement, that $H_{t-1}$ must stay in the range $[\theta_t \pm c_t]$ in a significant part of the time fails if we scale the process. This problem can be handled if we scale somehow the steps of the algorithm. Since the problem is in the step function $H_{t-1}\in[\theta_t \pm c_t]$,
    the steps $c_t$ has to reflect the scale of the process ($H_{t-1}$ and $\theta_t$ are on the same scale). If we re-scale the $c_t$ variable then we need to compensate the $a_t/c_t$ term as well.
    Therefore the steps are the following:
    \begin{equation}
        a_t = K t^{-p}, \quad c_t = K t^{-q},
    \end{equation}
    where $K$ equals to the standard deviation of $H_t$.
    It could be an estimation of the standard deviation but for simplicity we used the whole dataset to estimate it.
    
    \begin{figure}
        \centering
        \begin{subfigure}{.5\textwidth}
          \centering
          \includegraphics[width=1\linewidth]{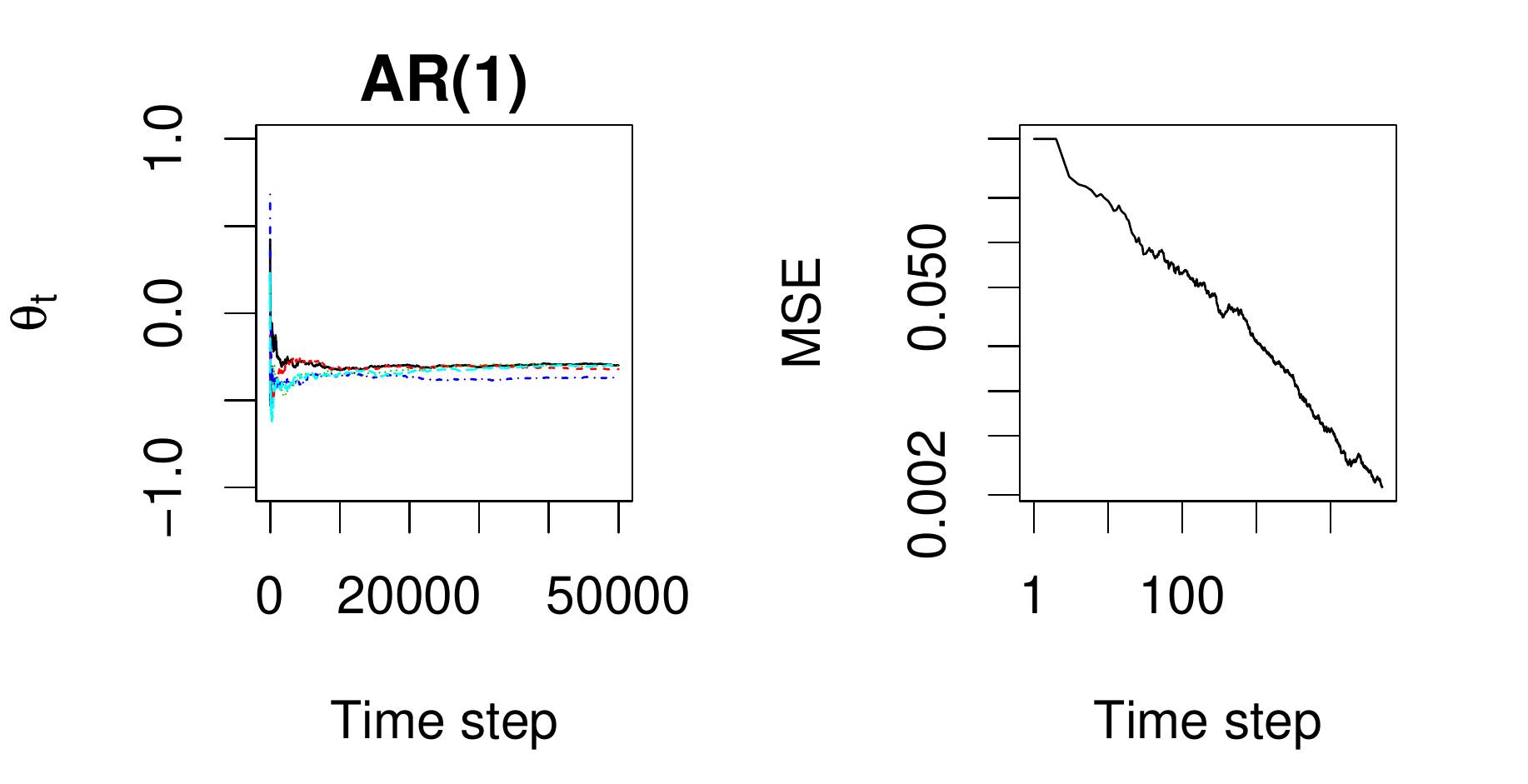}
          \caption{A subfigure}
          \label{fig:kwar1}
        \end{subfigure}%
        \begin{subfigure}{.5\textwidth}
          \centering
          \includegraphics[width=1\linewidth]{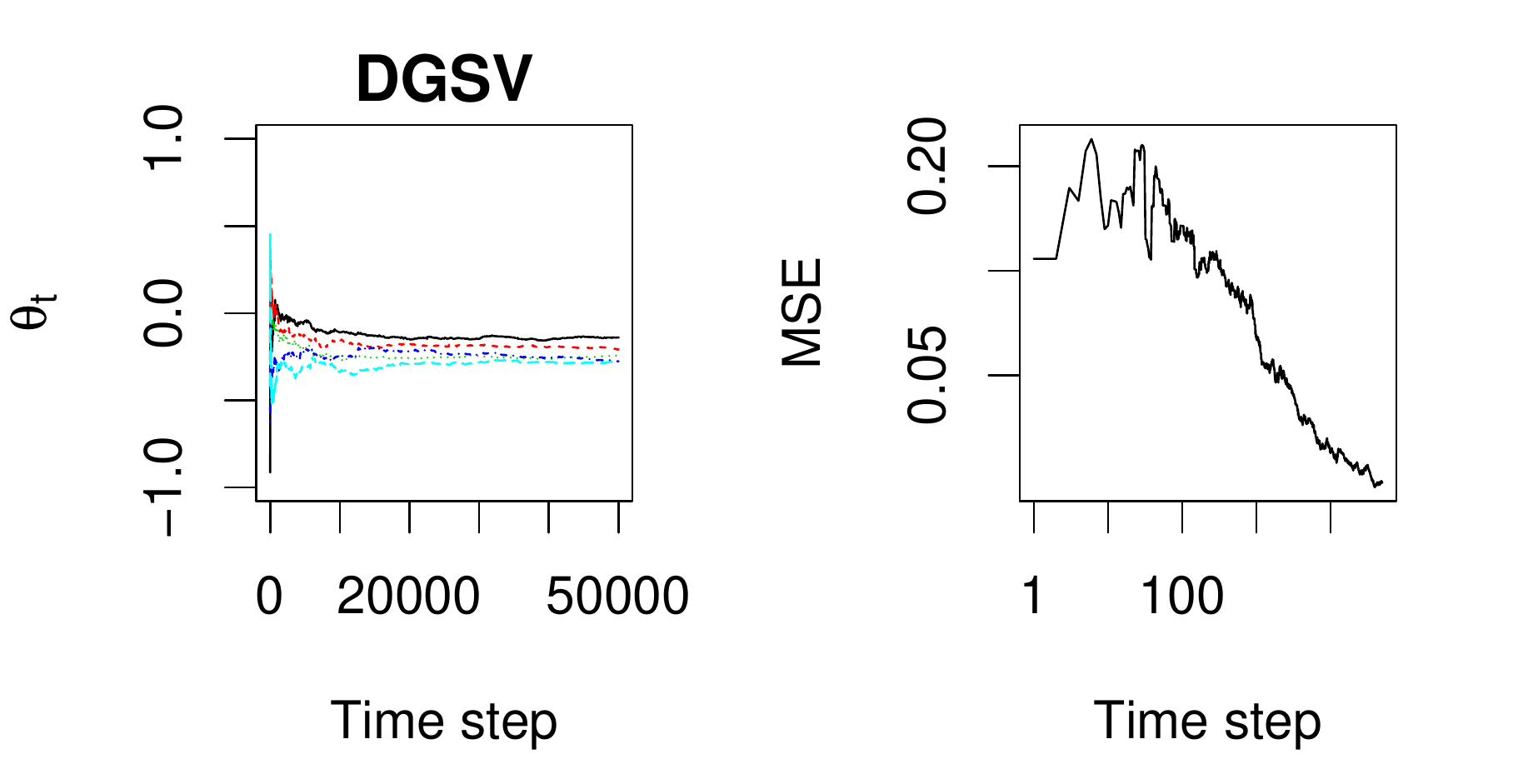}
          \caption{A subfigure}
          \label{fig:kwdgsv}
        \end{subfigure}
        \caption{Convergence of the Kiefer--Wolfowitz algorithm in $L^2$ for the some process as in \ref{fig:hill}. The left figures shows 5-5 realizations, while the right figures shows the Mean Squared Error (MSE) of all the realizations.   }
        \label{fig:kw}
    \end{figure}
    
    The performance of using the scaling factor $K$ on both $a_t$ and $c_t$ is showed on Table \ref{tab:scaling}!!! and
    Figure !!!. The table shows the Mean Squared Error at $t=T=100\,000$, while the figure shows the function $t\rightarrow MSE_t$ with different scaling.
    Parameter settings of the table and the figure is defined below in (\ref{eq:data1}) and (\ref{eq:data2}).
    In Dataset-2, when $\alpha$ is smaller, the Mean Squared Error is higher despite that the process's variation is higher (in the AR(1) case the variation is $\sigma^2/(1-\alpha^2) $). 
    This is because the lower the $\alpha$ the less information we have, it is more difficult to learn.
    Figures \ref{fig:scalingar1} and \ref{fig:scalingdgsv} show thatwithout scaling the algorithm at first wait until $c_t$ achieves a suitable size, while using scaling speeds up this and the algorithm uses the appropriate $c_t$'s.
    The numerical results also show that the best way to scale the process is using the fivefold of the standard deviation of the process.
    
    \noindent Dataset 1:
    \begin{equation}\label{eq:data1}
        \mu = 0.01,\, \alpha = 0.5,\, \sigma = 0.05,\, \rho = -0.2,\, b_0 = 0.4,\, b = 0.7.
    \end{equation}
    Dataset 2:
    \begin{equation}\label{eq:data2}
        \mu = 0.005,\, \alpha = 0.2,\, \sigma = 0.05,\, \rho = -0.2,\, b_0 = 0.4,\, b = 0.7.
    \end{equation}
    (In the AR(1) only $\mu,\alpha,\sigma$ make sense.)
    \begin{table}[!htb]
			\begin{center}
				\begin{tabular}{@{}lccccc@{}}
					\toprule
					Scaling &  \multicolumn{2}{c}{AR(1)}& &  \multicolumn{2}{c}{DGSV} \\
		         \cline{2-3}  \cline{5-6}\\
					  &  Dataset-1 ($\times 10^{-6}$)  & Dataset-2 ($\times 10^{-5}$) & & Dataset-1 ($\times 10^{-6}$) & Dataset-2 ($\times 10^{-5}$)\\
 				 \cline{2-6}\\
		No scaling               & 7.8  &  1.8  && 6.7 & 19.2 \\
		$K = $st. dev.$(H_t)$    & 11.4 &  17.9 && 53  & 120.0\\
		$K = $st. dev.$(H_t)5$   & 1.7  &  1.6  && 20  & 8.8  \\
				 \bottomrule
				\end{tabular}
			\end{center}
			\caption{Performance of the algorithm with different scaling of the steps $a_t$ and $c_t$. }
			\label{tab:scaling}
		\end{table}
    \begin{figure}
        \centering
        \includegraphics[width=1\linewidth]{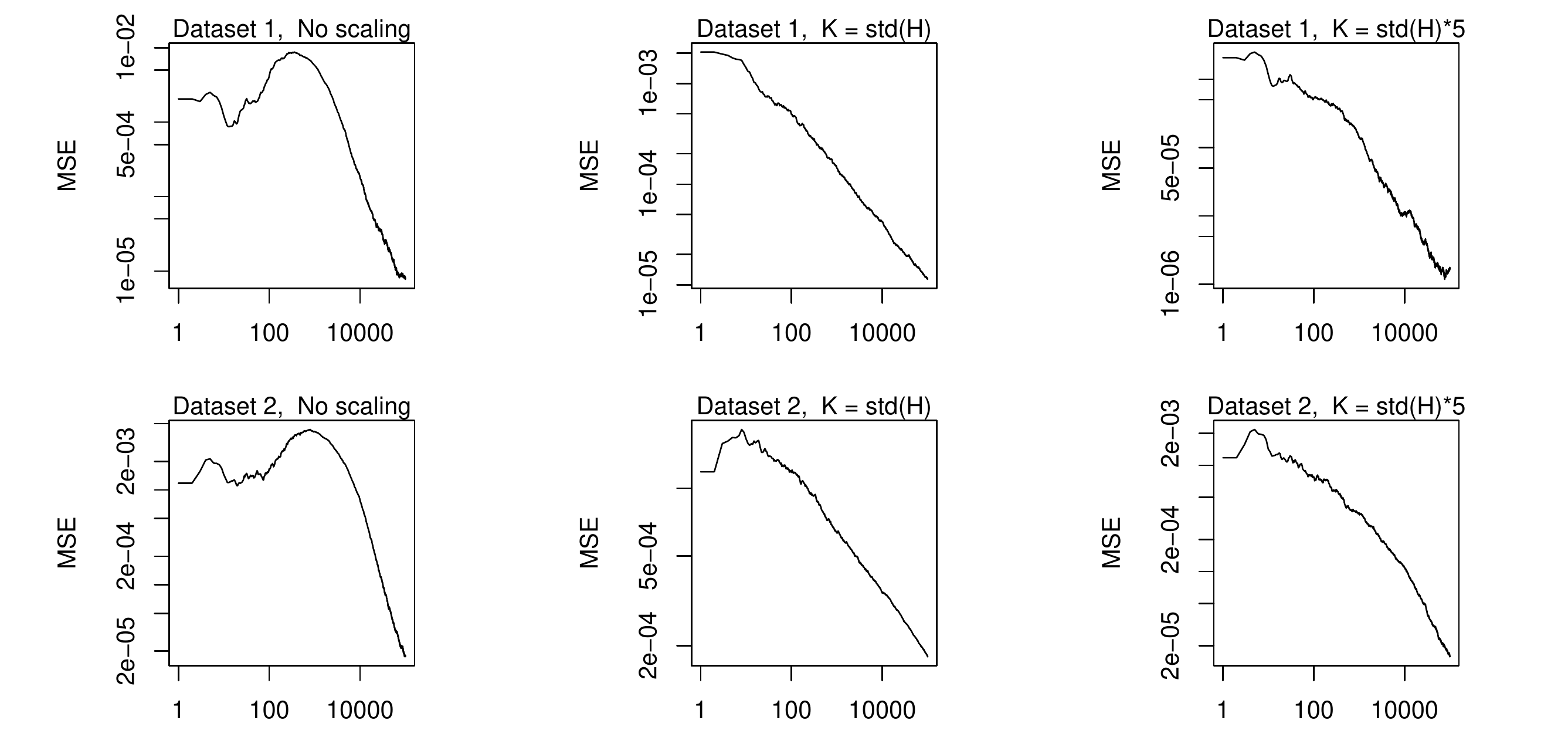}
        \caption{Investigating the effect of different scaling for two different parametrization of AR(1). For the parameters see (\ref{eq:data1}) and (\ref{eq:data2}), for the outcome see Table \ref{tab:scaling}.}
        \label{fig:scalingar1}
    \end{figure}
    \begin{figure}
        \centering
        \includegraphics[width=1\linewidth]{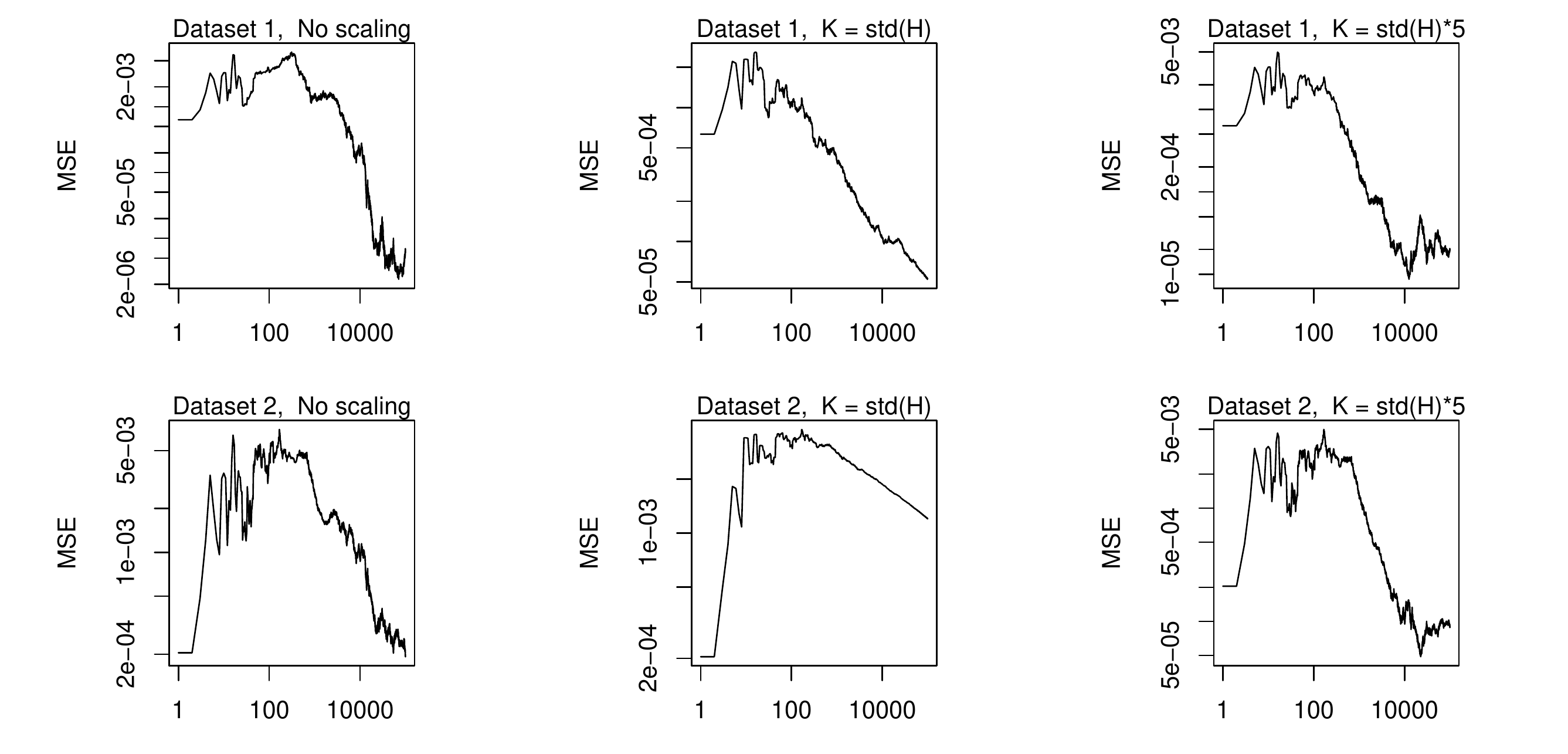}
        \caption{Investigating the effect of different scaling for two different parametrization of DGSV. For the parameters see (\ref{eq:data1}) and (\ref{eq:data2}), for the outcome see Table \ref{tab:scaling}.}
        \label{fig:scalingdgsv}
    \end{figure}

\section*{Funding}
The first author gratefully acknowledges the support of Új Nemzeti Kiválóság Program {2018/2019}, of Ministry of Human Capacities (project number: ÚNKP-18-3-IV-PPKE-21). The second author acknowledges support from the "Lendület" grant LP 2015-16 of the Hungarian Academy of Sciences (Lendület grant LM 2015-16) and supported by the NKFIH (National
Research, Development and Innovation Office, Hungary) grant KH 126505.

\bibliographystyle{unsrt}
\bibliography{bibliography}

\end{document}